\documentclass[a4paper]{article}

\usepackage{graphicx}
\graphicspath{{figs/}}
\usepackage[english]{babel}
\usepackage{latexsym}
\usepackage{url}
\usepackage{amssymb}
\usepackage{amsmath}
\usepackage{amsthm}
\usepackage{amsfonts}
\usepackage{ifthen}

\usepackage{fullpage}

\newtheorem{theorem}{Theorem}
\newtheorem{corollary}[theorem]{Corollary}
\newtheorem{lemma}[theorem]{Lemma}

\newcommand{\RR}{\ensuremath{\mathbb R}}  % real numbers
\newcommand{\U}{\ensuremath{\mathcal{U}}}  
\newcommand{\scalar}[1]{\ensuremath{\langle #1 \rangle}} 

\DeclareMathOperator{\opt}{opt}
\DeclareMathOperator{\area}{area}
\DeclareMathOperator{\dwidth}{dwidth}
\DeclareMathOperator{\peri}{peri}

\newcommand\eps{\varepsilon}
\newcommand{\paral}{\lozenge}
\newcommand{\Ropt}{R_{\opt}}

\def\DEF#1{\textbf{\emph{#1}}}
 
\newenvironment{di}{\list{$\bullet$}{\itemsep0pt\parsep0pt}}{\endlist}

\let\leq\leqslant
\let\geq\geqslant
\let\le\leqslant
\let\ge\geqslant

\makeatletter
\ifx\showkeyslabelformat\@undefined\else
\renewcommand{\showkeyslabelformat}[1]{\normalfont\tiny\ttfamily#1}
\fi
\partopsep\z@ \textfloatsep 10pt plus 1pt minus 4pt
\def\section{\@startsection {section}{1}{\z@}{-3.5ex plus -1ex minus
    -.2ex}{2.3ex plus .2ex}{\large\bf}}
\def\subsection{\@startsection{subsection}{2}{\z@}{-3.25ex plus -1ex
    minus -.2ex}{1.5ex plus .2ex}{\normalsize\bf}}
\def\@fnsymbol#1{\ensuremath{\ifcase#1\or *\or 1\or 2\or 3\or 4\or
    5\or 6\or 7 \or 8\ or 9 \or 10\or 11 \else\@ctrerr\fi}}
\makeatother

\newcommand{\out}[1]{}

\title{Finding Largest Rectangles in Convex Polygons%
	\thanks{Supported by the Slovenian Research Agency, program P1-0297, projects J1-4106 and L7-5459; by the ESF EuroGIGA project (project GReGAS) of the European Science Foundation; and by NRF
    grant~2011-0030044 (SRC-GAIA), funded by the government of Korea.}}

\author{Sergio Cabello%
	\thanks{Department of Mathematics, IMFM, and
                Department of Mathematics, FMF, University of Ljubljana, Slovenia.
                Part of the work was done while visiting KAIST, Korea.}
  \and Otfried Cheong% 
  	\thanks{Department of Computer Science, KAIST, Daejeon, Korea.}
  \and Christian Knauer%
  	\thanks{Institut f\"ur Informatik, Universit\"at Bayreuth, Bayreuth, Germany}
  \and Lena Schlipf%
  	\thanks{Institute of Computer Science, Freie Universit\"at Berlin,
    Germany.}
}

\begin{document}
\maketitle

\begin{abstract}
  We consider the following geometric optimization problem: find a
  maximum-area rectangle and a maximum-perimeter rectangle contained
  in a given convex polygon with $n$ vertices.  We give exact
  algorithms that solve these problems in time $O(n^3)$.  We also give
  $(1-\eps)$-approximation algorithms that take time $O(\eps^{-3/2}+
  \eps^{-1/2} \log n)$.
	
    \medskip
    \textbf{Keywords:} geometric optimization; approximation algorithm;
    convex polygon; inscribed rectangle.
\end{abstract}

%%%%%%%%%%%%%%%%%%%%%%%%%%%%%%%%%%%%%%%%%%%%%%%%%%%%%%%%%%%%%%%%%%%%%%%%%%%%%%%%%%%%%%%%%%%%%%%%%%%%%%%%%%%%%%%%%%%%%%
\section{Introduction}

Computing a largest rectangle contained in a polygon (with respect to
some appropriate measure) is a well studied problem.  Previous results
include computing largest \emph{axis-aligned} rectangles, either in
convex polygons~\cite{ahs-cliir-95} or simple polygons (possibly with
holes) \cite{DMR-1997}, and computing largest \emph{fat} rectangles in
simple polygons \cite{hkms-06}.

Here we study the problem of finding a maximum-area rectangle and a
maximum-perimeter rectangle contained in a given convex polygon with
$n$ vertices.  We give exact $O(n^3)$-time algorithms and
$(1-\eps)$-approximation algorithms that take time $O(\eps^{-3/2}+
\eps^{-1/2} \log n)$. (For maximizing the perimeter we allow the
degenerate solution consisting of a single segment whose perimeter is
twice its length.)  To the best of our knowledge, apart from a
straightforward $\mathcal O(n^4)$-time algorithm, there is no other
exact algorithm known so far.

Our approximation algorithm to maximize the area improves the previous
results by Knauer et al.~\cite{ksst-12}, who give a deterministic
$(1-\eps)$-approximation algorithm with running time $O(\eps^{-2}\log
n)$ and a Monte Carlo $(1-\eps)$-approximation algorithm with running
time $O(\eps^{-1}\log n)$. We are not aware of previous
$(1-\eps)$-approximation algorithms to maximize the perimeter.

%%%%%%%%%%%%%%%%%%%%%%%%%%%%%%%%%%%%%%%%%%%%%%%%%%%%%%%%%%%%%%%%%%%%%%%%%%%%%%%%%%%%%%%%%%%%%%%%%%%%%%%%%%%%%%%%%%%%%%
\section{Preliminaries}

\paragraph{Notation.}
We use $C$ for arbitrary convex bodies and $P$ for convex polygons.

Let $\U$ be the set of unit vectors in the plane.
For each $u\in \U$ and each convex body $C$, the \emph{directional width}
of $C$ in direction $u$, denoted by $\dwidth (u,C)$,
is the length of the orthogonal projection of $C$ 
onto any line parallel to $u$. Thus
\[
	\dwidth (u,C) ~=~ \max_{p\in C} \scalar{p,u} - \min_{p\in C} \scalar{p,u},
\]
where $\scalar{\cdot,\cdot}$ denotes the scalar product.

For a convex body $C$ and a parameter $\eps\in (0,1)$, 
an \emph{$\eps$-kernel} for $C$ is a convex body $C_\eps\subseteq C$ such that
\[
	\forall u\in \U: ~~(1-\eps)\cdot \dwidth (u,C) \le \dwidth (u,C_\eps).
\]
The diameter of $C$ is the distance between the two furthest points
of $C$. It is easy to see that it equals
\[
	\max_{u\in \U} \dwidth (u,C).
\]

Ahn et al.~\cite{abcnsv-06} show how to compute an $\eps$-kernel.
Their algorithm uses the following type of 
primitive operations for $C$:
\begin{di}
	\item given a direction $u\in \U$, find an extremal point of $C$
		in the direction $u$; 
	\item given a line $\ell$, find $C\cap \ell$.
\end{di}
Let $T_C$ be the time needed to perform each of those primitive operations.
We will use $T_C$ as a parameter in some of our running times.
When $C$ is a convex $n$-gon whose boundary is given 
as a sorted array of vertices or as a binary search tree,
we have $T_C=O(\log n)$~\cite{Chazelle-Dobkin,Reichling}.
Ahn et al.~show the following result.
\begin{lemma}[Ahn et al.~\cite{abcnsv-06}]
\label{le:kernel}
	Given a convex body $C$ and a parameter $\eps\in (0,1)$,
	we can compute in $O(\eps^{-1/2} T_C)$ time
	an $\eps$-kernel of $C$ with $O(\eps^{-1/2})$ vertices.
\end{lemma}

\begin{lemma}
\label{le:invariant}
	Let $C_\eps$ be an $\eps$-kernel for $C$.
	If $\varphi$ is an invertible affine mapping,
	then $\varphi(C_\eps)$ is an $\eps$-kernel for $\varphi(C)$.
\end{lemma}
\begin{proof}
	The ratio of directional widths for convex bodies
	is invariant under invertible affine transformations.
	This means that  
	\[
		\forall u\in \U: ~~ 
		1-\eps 
		\le \frac{\dwidth (u,C_\eps)}{\dwidth (u,C)}
		=\frac{\dwidth (u,\varphi(C_\eps))}{\dwidth (u,\varphi(C))}
	\]
	and thus $\varphi(C_\eps)$ is an $\eps$-kernel for $\varphi(C)$.
\end{proof}

\begin{figure}
	\includegraphics[width=\columnwidth,page=1]{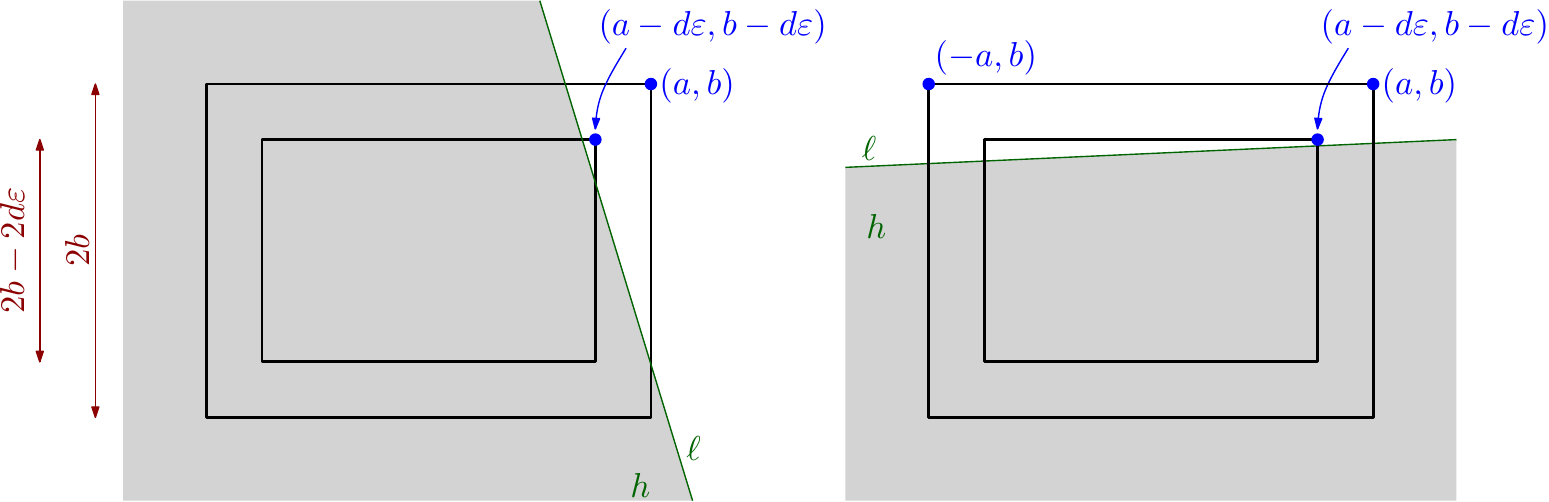}
    \caption{Proof of Lemma~\ref{le:square}.}
    \label{fig:square}
\end{figure}

\begin{lemma}
  \label{le:square}
  Assume that $C$ contains the rectangle $R = [-a,a] \times [-b,b]$,
  that $C$ has diameter $d$, and that $C_\eps$ is an $\eps$-kernel for
  $C$.  Then $C_\eps$ contains the axis-parallel rectangle $S = [-a +
    d\eps, a - d\eps] \times [-b + d\eps, b - d\eps]$.
\end{lemma}
\begin{proof}
  The statement is empty if $a < d\eps$ or $b < d\eps$, so assume that
  $a, b \geq d\eps$.  For the sake of contradiction, assume that $S$
  is not contained in $C_\eps$.  This means that one vertex of $S$ is
  not contained in $C_\eps$.  Because of symmetry, we can assume that
  $s=(a-d\eps, b-d\eps)$ is not contained in $C_\eps$. Since $C_\eps$
  is convex and $s\notin C_\eps$, there exists a closed halfplane $h$
  that contains $C_\eps$ but does not contain $s$.  Let $\ell$ be the
  boundary of $h$.
	
  We next argue that $R$ has some vertex at distance at least $d\eps$
  from $h$ (and thus $\ell$); see Figure~\ref{fig:square} for a couple
  of cases.  If $\ell$ has negative slope and $h$ is its lower
  halfplane, then the distance from $(a,b)$ to $\ell$ is at least
  $d\eps$.  If $\ell$ has negative slope and $h$ is its upper
  halfplane, then the distance from $(-a,-b)$ to $\ell$ is at least
  $2b-d\eps \geq d\eps$.  If $\ell$ has positive slope, then $(-a,b)$
  or $(a,-b)$ are at distance at least $d\eps$ from~$h$.
	
  Since $R\subseteq C$, for the direction $u$ orthogonal to~$\ell$ we
  have
  \[
  \dwidth(u,C) - \dwidth(u,C_\eps) > 
  d\eps \ge \eps \cdot \dwidth(u,C),		
  \]
  where we have used the assumption that $\dwidth(u,C)\le d$.
  This means that 
  \[
  \left( 1 - \eps \right)\cdot \dwidth(u,C) ~>~
  \dwidth(u,C_\eps),
  \]
  which contradicts that $C_\eps$ is an $\eps$-kernel
  for $C$.
\end{proof}

%%%%%%%%%%%%%%%%%%%%%%%%%%%%%%%%%%%%%%%%%%%%%%%%%%%%%%%%%%%%%%%%%%%%%%%%%%%%%%%%%%%%%%%%%%%%%%%
\section{Exact algorithms}
\label{sec:exact}

Let $e_1,\dots, e_n$ be the edges of the convex polygon $P$.
For each edge $e_i$ of $P$, let $h_i$ be the closed halfplane defined
by the line supporting $e_i$ that contains $P$.
Since $P$ is convex, we have $P=\bigcap_i h_i$.

We parameterize the set of \emph{parallelograms} in the plane by points 
in $\RR^6$, as follows.
We identify each $6$-dimensional point $(x_1,x_2,u_1,u_2,v_1,v_2)$
with the triple $(x,u,v)\in (\RR^2)^3$, where $x=(x_1,x_2)$, 
$u=(u_1,u_2)$, and $v=(v_1,v_2)$. The triple $(x,u,v)\in \RR^6$
corresponds to the parallelogram $\paral(x,u,v)$ with vertices
\[
	x,~ x+u,~ x+v,~ x+u+v.
\]
Thus, $x$ describes a vertex of the parallelogram $\paral(x,u,v)$,
while $u$ and $v$ are vectors describing the edges of $\paral(x,u,v)$.
This correspondence is not bijective because, for example,
\[
	\paral(x,u,v) ~=~ \paral(x+u+v,-u,-v) ~=~ 
	\paral(x,v,u) .
\] 
Nevertheless, each parallelogram is $\paral(x,u,v)$ for
some $(x,u,v)\in \RR^6$: the parallelogram given by 
the vertices $p_1p_2p_3p_4$
in clockwise (or counterclockwise) order 
is $\paral (p_1,p_2-p_1,p_4-p_1)$.

We are interested in the parallelograms contained in $P$.
To this end we define 
\[
	\Pi(P) ~=~ \big\{ (x,u,v)\in \RR^6 \mid \paral(x,u,v) \subseteq P\big\}.
\]
Since $P$ is convex, a parallelogram is contained in $P$ if and
only if each vertex of the parallelogram is in~$P$.
Therefore 
\begin{align*}
	\Pi(P) ~&=~ \big\{ (x,u,v)\in \RR^6 \mid x,\, x+u,\, x+v,\, x+u+v \in P\big\} \\
	&=~ \big\{ (x,u,v)\in \RR^6 \mid \forall i: x,\, x+u,\, x+v,\, x+u+v \in h_i \big\} \\
	&=~ \bigcap_i \big\{ (x,u,v)\in \RR^6 \mid x,\, x+u,\, x+v,\, x+u+v \in h_i \big\}.
\end{align*}
Since $\Pi(P)$ is trivially bounded, 
it follows that $\Pi(P)$ is a convex polytope in $\RR^6$ 
defined by $4n$ linear constraints. 
The Upper Bound Theorem~\cite{McMullen} implies that $\Pi(P)$ has combinatorial
complexity at most~$O(n^3)$. Chazelle's algorithm~\cite{chazelle93}
gives a triangulation of the boundary of $\Pi(P)$
in $O(n^3)$ time; 
Seidel~\cite{seidel} calls this the boundary description of a polytope. 
From the triangulation of the boundary
we can construct a triangulation of $\Pi(P)$: chose
an arbitrary vertex $x$ of $\Pi(P)$ and add it to 
each simplex of the triangulation of the boundary of $\Pi(P)$ that does not contain $x$.
(One can also use a point in the interior of $\Pi(P)$.)

The set of \emph{rectangles} is obtained by restricting
our attention to triples $(x,u,v)$ with $\scalar{u,v}=0$, 
where $\scalar{\cdot,\cdot}=0$ again denotes the scalar product of two vectors.
This constraint is non-linear. Because of this, it
is more convenient to treat each simplex of a triangulation of
$\Pi(P)$ separately. When $\scalar{u,v}=0$, 
the area of $\paral(x,u,v)$ is $|u|\cdot |v|$.

Consider any simplex $\triangle$ of the triangulation of $\Pi(P)$. 
Finding the maximum area rectangle restricted to $\triangle$ 
corresponds to the problem
\begin{align*}
	\opt(\triangle) ~=~\max 		~& |u|^2\cdot |v|^2\\
		\text{s.t.} ~& (x,u,v)\in \triangle\\
				 & \scalar{u,v}=0
\end{align*}
This is a constant-size problem. It has $6$ variables and a constant
number of constraints; all constraints but one are linear.
The optimization function has degree four.
In any case, each such problem can be solved in constant time.
When the problem is not feasible, we set $\opt(\triangle)=0$.

Taking the best rectangle over all simplices of a triangulation of
$\Pi(P)$, we find a maximum area rectangle.
Thus, we return $\arg\max_{\triangle} \opt(\triangle)$.
We have shown the following.
\begin{theorem}
\label{th:exact}
	Let $P$ be a convex polygon with $n$ vertices. 
	In time $O(n^3)$ we can find a maximum-area rectangle
	contained in $P$.
\end{theorem}

To maximize the perimeter, we apply the same approach.
For each simplex $\triangle$ in a triangulation of $\Pi(P)$
we have to solve the following problem:
\begin{align*}
	\opt(\triangle) ~=~\max 		~& |u| + |v|\\
		\text{s.t.} ~& (x,u,v)\in \triangle\\
				 & \scalar{u,v}=0
\end{align*}
Combining the solutions over all simplices of the triangulation
we obtain the following.
\begin{theorem}
\label{th:exact2}
	Let $P$ be a convex polygon with $n$ vertices. 
	In time $O(n^3)$ we can find a maximum-perimeter rectangle
	contained in $P$.
\end{theorem}

\section{Combinatorially distinct rectangles}

The reader may wonder if the algorithm of the previous section cannot be improved:
It constructs the space~$\Pi(P)$ of all \emph{parallelograms} contained in~$P$, 
and then considers the intersection with the manifold $\scalar{u,v} = 0$ corresponding to 
the rectangles.  If the complexity of this intersection was smaller than~$\Theta(n^3)$, 
then we should avoid constructing the entire parallelogram space~$\Pi(P)$ first.

In this section we show that this is not the case: the complexity of the space of rectangles
that fit inside~$P$, that is, the complexity of the intersection of~$\Pi(P)$ with the 
manifold~$\scalar{u, v} = 0$, is already~$\Theta(n^3)$ in the worst case. Therefore,
asymptotically we are not loosing anything by considering all parallelograms, 
instead of directly concentrating on rectangles.

To this end, let us call two rectangles contained in a convex polygon~$P$ 
\DEF{combinatorially distinct} 
if their vertices are incident to a different subset of edges of~$P$.
We are going to show the following: for every sufficiently large $n$ there is a 
polygon $P$ with $n$ vertices that contains $\Theta(n^3)$ combinatorially
distinct rectangles. This shows that any algorithm iterating over all
combinatorially distinct rectangles contained in $P$ needs at least
$\Omega(n^3)$ time. Our algorithm falls in this category.

We provide an informal overview of the construction, see Figure~\ref{fig:lowerbound1}.
\begin{figure}
	\centering
	\includegraphics[scale=.8,page=1]{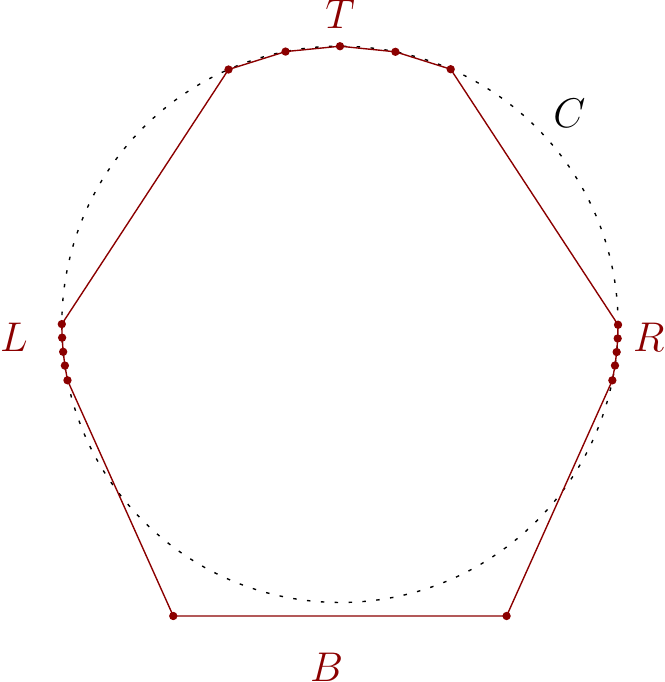}
    \caption{Outline of the construction.}
    \label{fig:lowerbound1}
\end{figure}
For simplicity, we are going to use $3n+2$ vertices. 
Consider the circle $C$ of unit radius centered at the origin.
We are going to select some points on $C$, plus two additional points. 
The polygon $P$ is then described as the convex hull of these points.
The points are classified into 4 groups. 
We have a group $L$ of $n$ points
placed densely on the left side of $C$.
We have another group $R$ of $n$ points placed densely on the right 
side of $C$.
The third group $T$, also with $n$ points, is placed on the upper part of $C$.
The points of $T$ are more spread out and will be chosen carefully.
Finally, we construct a group $B$ of two points placed below $C$. The construction
will have the property that, for any edge $e_L$ defined by $L$, any edge $e_R$
defined by $R$, and any edge $e_T$ defined by $T$, there is a rectangle
contained in $P$ with vertices on the edges $e_L$, $e_R$ and $e_T$. The points
$B$ are needed only to make sure that the bottom part of the rectangle
is contained in $P$; they do not have any particular role.

\begin{theorem}
	For any sufficiently large value of $n$ there is a polygon $P$ with
	$n$ vertices such that $P$ contains $\Theta(n^3)$ combinatorially distinct 
	rectangles.
\end{theorem}
\begin{proof}
	Let $C$ be the circle of unit radius centered at the origin~$o$,
	and let $C'$ be the circle of radius $1-2\eps$ centered at~$o$, for a 
    small value~$\eps>0$ to be chosen later on.
	Let $H_\eps$ be the horizontal strip defined by $-\eps\le y\le 0$, 
    see Figure~\ref{fig:lowerbound2}.
	Select a set $R$ of $n$ points in $C\cap H_\eps$ with positive $x$-coordinate
	and select a set $L$ of $n$ points in $C\cap H_\eps$ with negative $x$-coordinate.
	For every $r$ on a segment connecting consecutive points of $R$ 
	and every $\ell$ on a segment connecting consecutive points of $L$, 
	let $C_{r\ell}$ be the circle with diameter $r\ell$.  
	We now observe that the upper semicircle of $C_{r\ell}$ with endpoints~$r$ and~$\ell$ 
    lies between $C$ and $C'$.
    This follows 
	from the fact that $C_{r\ell}$ has radius at least $1-\eps$ and the center
	of $C_{r\ell}$ is at most $\eps$ apart from~$o$.
\begin{figure}
	\centering
	\includegraphics[scale=.6,page=2]{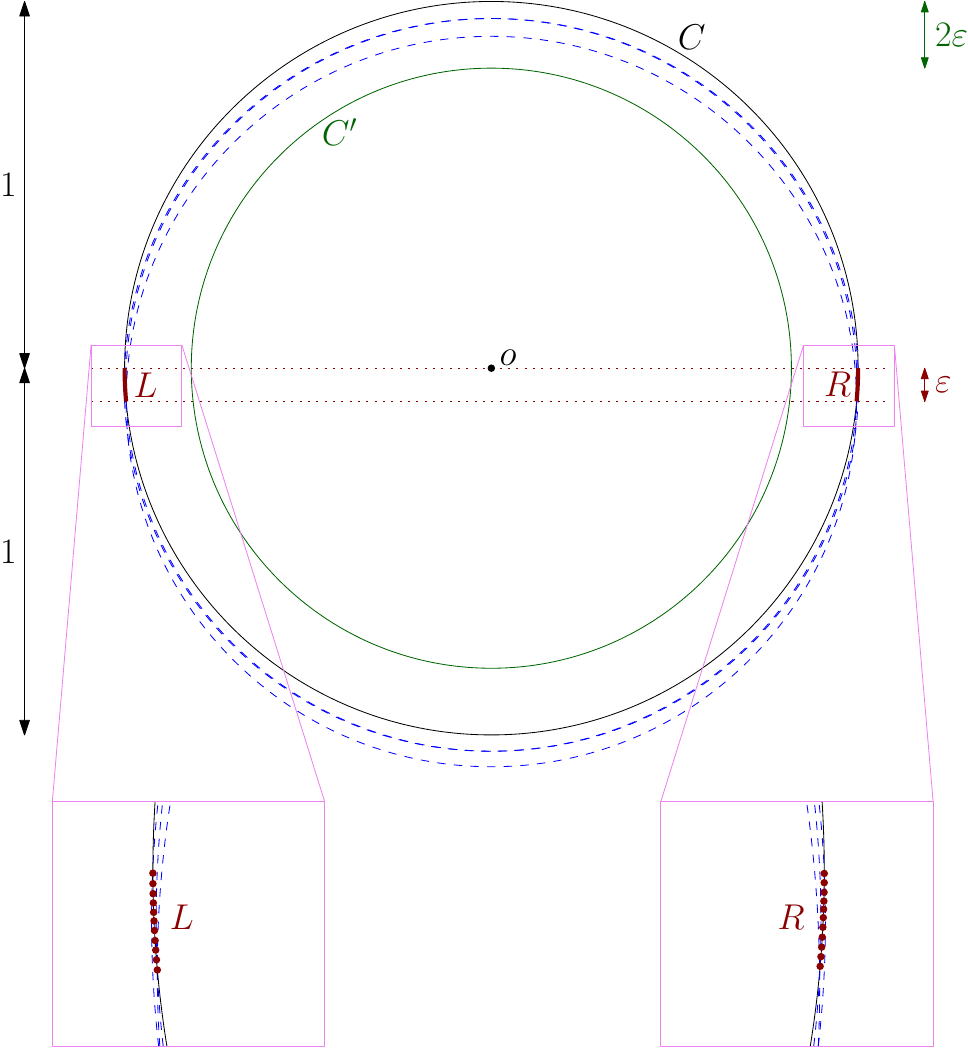}
    \caption{Detail of the construction of $L$ and $B$. Some circles $C_{r\ell}$
		are in dashed blue.}
    \label{fig:lowerbound2}
\end{figure}

	Figure~\ref{fig:lowerbound3} should help with the continuation of the construction.
	We place a set $T$ of $n$ points on the upper side of $C$. We want the following
	additional property: for any two consecutive points $t$ ant $t'$ of $T$,
	the segment $tt'$ intersects $C'$. If we select $\eps>0$ sufficiently small,
	then $C$ and $C'$ are close enough that we can select the $n$ points needed
	to construct $T$. Finally, we choose a set $B$ of two points below $C$, as shown in 
	Figure~\ref{fig:lowerbound1}. The final polygon $P$ is the convex hull
	of $L\cup R\cup T\cup B$. This finishes the description of the polygon $P$.
	
	Consider any edge $e_R$ defined by two consecutive vertices of $R$ 
	and chose a point $r$ on $e_R$.
	Similarly, consider any edge $e_L$ defined by two consecutive vertices of $L$ 
	and chose a point $\ell$ on $e_L$.
	Let $e_T$ be an edge of $P$ defined by two consecutive points of $T$. 
	By construction, the circle $C_{r\ell}$ with diameter $r\ell$
	intersects the segment $e_T$ in some point, let's call it~$p$. 
	This means that the triangle $\triangle(r,p,\ell)$ has a
    right angle at~$p$.  Let~$q$ be the point on~$C_{r\ell}$ such that $pq$ is a diameter
    of~$C_{r\ell}$.  Then the quadrilateral formed by $r$, $p$, $\ell$, and~$q$ 
    is a rectangle contained in~$P$.
	Each choice of an edge $e_R$ defined by $R$, an edge $e_L$ defined by $L$, 
	and an edge $e_T$ defined by $T$ results in a combinatorially distinct rectangle.
	Therefore there are $\Omega(n^3)$ combinatorially distinct rectangles contained in~$P$.	
\end{proof}

\begin{figure}
	\centering
	\includegraphics[scale=.6,page=3]{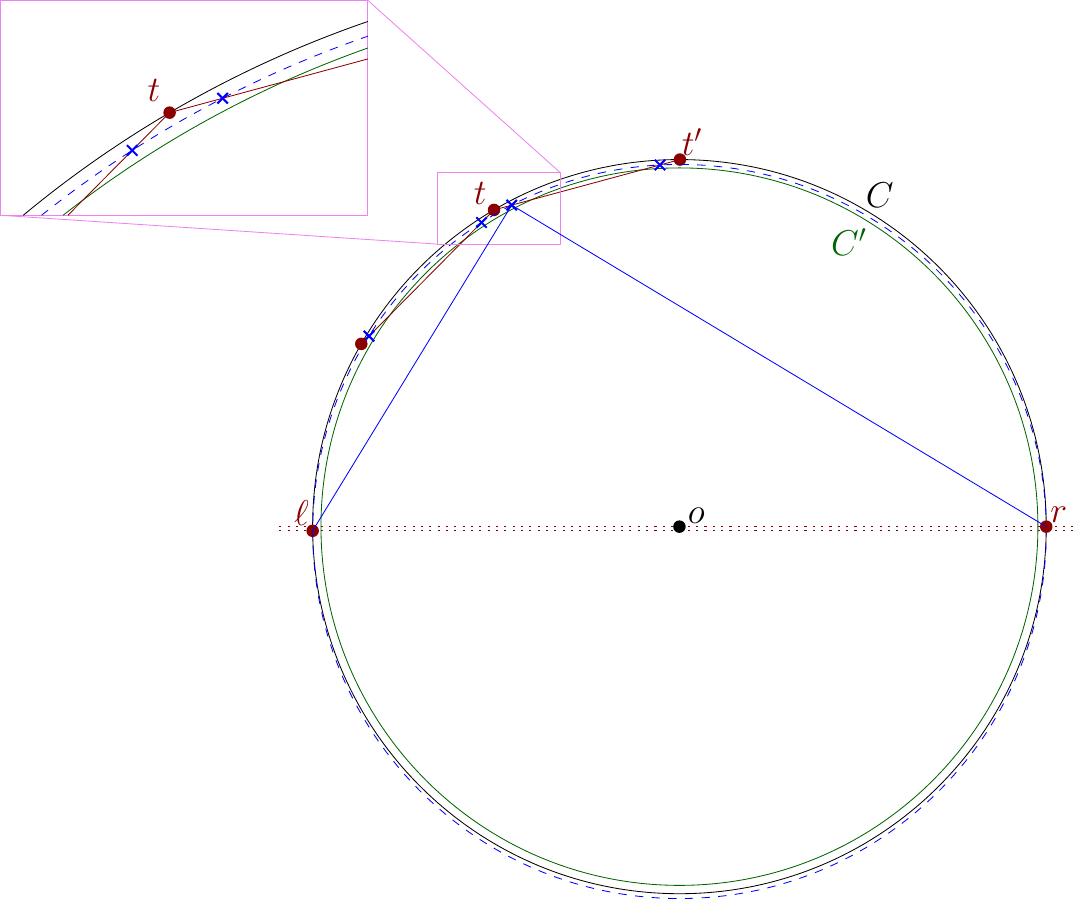}
    \caption{Detail of the construction of $T$.}
    \label{fig:lowerbound3}
\end{figure}

%%%%%%%%%%%%%%%%%%%%%%%%%%%%%%%%%%%%%%%%%%%%%%%%%%%%%%%%%%%%%%%%%%%%%%%%%%%%%%%%%%%%%%%%%%%%%%%%%%%%%%%%%%%%%%%%%%%%%%%%%%%%%%%%%%%%%%%%%%%%%%%%
\section{Approximation algorithm to maximize the area}
\label{sec:approximation1}

The algorithm is very simple: we compute an $(\eps/32)$-kernel 
$C_{\eps/32}$ for the input convex body $C$,
compute a maximum-area rectangle contained in $C_{\eps/32}$ and return it.
We next show that this algorithm indeed returns a $(1-\eps)$-approximation.

Let $\Ropt$ be a maximum-area rectangle contained in $C$, and
let $\varphi$ be an affine transformation such that
$\varphi(\Ropt)$ is the square $[-1,1]^2$. 

\begin{lemma}
\label{le:bound}
	The diameter of $\varphi(C)$ is at most $16$.
\end{lemma}
\begin{proof}
	We will show that $\varphi(C)$ is contained in the disk 
	centered at	the origin $o=(0,0)$ of radius $8$, which implies the result.

	Any convex body contains a rectangle with at least half of its 
	area~\cite{Lassak}. Therefore
	$\area(\Ropt)/\area(C)\ge 1/2$.
	
	Any invertible affine transformation does not change the ratio
	between areas of objects. Therefore
	\[
		\frac 12 ~\le~ \frac{\area(\Ropt)}{\area(C)} ~=~ 
		\frac{\area(\varphi(\Ropt))}{\area(\varphi(C))} ~=~
		\frac{4}{\area(\varphi(C))}
	\]
	and thus $\area(\varphi(C))\le 8$.
		
	\begin{figure}
		\centering
		\includegraphics[]{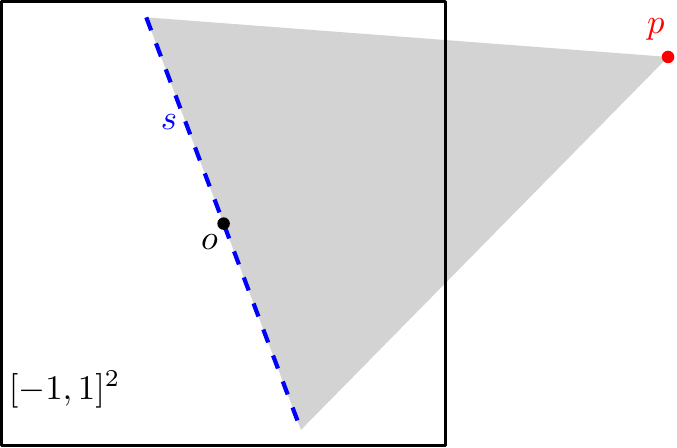}
		\caption{Triangle considered in the proof of Lemma~\ref{le:bound}. The point $p$ is
			drawn closer to the origin $o$ than it is assumed.}
		\label{fig:bound}
	\end{figure}

	Assume, for the sake of reaching a contradiction, that
	$\varphi(C)$ has a point $p$ at distance larger 
    than $8$ from the origin $o$. See Figure~\ref{fig:bound}.
	Let $s$ be the line segment
    of length $2$ centered at the origin $o$ and orthogonal
    to the segment $op$. Since $s$ is contained in the square $[-1,1]^2$,
    it is contained in $\varphi(C)$. Therefore  $\varphi(C)$ contains
	the convex hull of $s\cup \{ p \}$, which is a triangle of area
	larger than $8$, and we get a contradiction. 
	It follows that $\varphi(C)$ is contained
	in a disk centered at the origin $0$ of radius $8$.	
\end{proof}

\begin{lemma}
\label{le:large}
	Let $C_\eps$ be an $\eps$-kernel for $C$. 
	Then $C_\eps$ contains a rectangle with area at least 
	$(1-32\eps)\cdot \area(\Ropt)$.
\end{lemma}
\begin{proof}
  Because of Lemma~\ref{le:invariant}, $\varphi(C_\eps)$ is an
  $\eps$-kernel for~$\varphi(C)$.  Since $\varphi(C)$
  contains~$[-1,1]^2$ and has diameter at most~$16$ due to
  Lemma~\ref{le:bound}, Lemma~\ref{le:square} with $a=b=1$ implies
  that $\varphi(C_\eps)$ contains the square $S=[-t,t]^2$,
  where~$t=1-16\eps$.
    
  Since $S$ is obtained by scaling $[-1,1]^2 = \varphi(\Ropt)$ by
  $1-16\eps$, its preimage $R = \varphi^{-1}(S)$ is obtained by
  scaling $\Ropt$ by~$1-16\eps$ about its center.  It follows that $R$
  is a rectangle with area
    \begin{align*}
    \area(R) ~=~ (1-16\eps)^2 \cdot \area(\Ropt) 
    ~\ge~ (1-32\eps) \cdot \area(\Ropt),
    \end{align*}
    and the lemma follows.
\end{proof}

\begin{theorem}
	Let $C$ be a convex body in the plane. 
	For any given $\eps\in (0,1)$, we can find a $(1-\eps)$-approximation
	to the maximum-area rectangle contained in $C$ 
	in time $O(\eps^{-1/2} T_C +\eps^{-3/2})$.
\end{theorem}
\begin{proof}
	First, we compute an $(\eps/32)$-kernel $C_{\eps/32}$ to $C$. 
	We then compute a maximum-area rectangle
	contained in $C_{\eps/32}$ and return it. This finishes the
	description of the algorithm.
	
	Because of Lemma~\ref{le:large}, $C_{\eps/32}$ contains
	a rectangle of area at least $(1-\eps)\cdot \area(\Ropt)$,
	where $\Ropt$ is a maximum-area rectangle contained in $C$.
	Therefore, the algorithm returns a $(1-\eps)$-approximation
	to the maximum-area rectangle.

	Computing $C_{\eps/32}$ takes time 
	$O((\eps/32)^{-1/2} T_C)=O(\eps^{-1/2} T_C)$ because
	of Lemma~\ref{le:kernel}. Since $C_{\eps/32}$
	has $O((\eps/32)^{-1/2})=O(\eps^{-1/2})$ vertices, finding a
	largest rectangle contained in $C_{\eps/32}$
	takes time $O(\eps^{-3/2})$ because of Theorem~\ref{th:exact}.
\end{proof}

\begin{corollary}
	Let $C$ be a convex polygon with $n$ vertices given 
	as a sorted array or a balanced binary search tree.
	For any given $\eps\in (0,1)$, we can find a $(1-\eps)$-approximation
	to the maximum-area rectangle contained in $C$ 
	in time $O(\eps^{-1/2} \log n +\eps^{-3/2})$.
\end{corollary}
\begin{proof}
	In this case $T_C=O(\log n)$.
\end{proof}

%%%%%%%%%%%%%%%%%%%%%%%%%%%%%%%%%%%%%%%%%%%%%%%%%%%%%%%%%%%%%%%%%%%%%%%%%%%%%%%%%%%%%%%%%%%%%%%%%%%%%%%%%%%%%%%%%%%%%%%%%%%%%%%%%%%%%%%%%%%%%%%%
\section{Approximation algorithm to maximize the perimeter}
\label{sec:approximation2}

The approximation algorithm is the following: we compute an
$(\eps/16)$-kernel $C_{\eps/16}$ for the input convex body $C$,
compute a maximum-perimeter rectangle $R_{\eps/16}$ contained in
$C_{\eps/16}$, and return it.  We next show that this indeed computes
a $(1-\eps)$-approximation.

Since the algorithm is independent of the coordinate axes, 
we can assume that the maximum-perimeter rectangle contained in~$C$ 
is an axis-parallel rectangle $\Ropt = [-a,a] \times [-b,b]$ with $b \leq a$.  
We distinguish two cases depending
on the aspect ratio $b/a \le 1$ of $\Ropt$.  When $b/a \le \eps/2$, then
the longest segment contained in $C$ is a good approximation
to~$\Ropt$.  When $b/a > \eps/2$, then $\Ropt$ is fat enough that we can
use Lemma~\ref{le:square} to obtain a large-perimeter rectangle.  
\begin{lemma}
  \label{le:bound3}
  We have
  $\peri(R_{\eps/16})\ge (1-\eps)\cdot \peri(\Ropt)$.
\end{lemma}
\begin{proof}
  We assume first that $b \le (\eps/2) a$.  This implies $\peri(\Ropt)
  \le 4(1 + \eps/2)a$.  Because $C$ contains~$\Ropt$, the directional
  width of $C$ in the horizontal direction is at least~$2a$.  The
  diameter of $C_{\eps/16}$ is therefore at least~$(1-\eps/16)2a$, and
  so $C_{\eps/16}$ contains a segment of perimeter at least
  $4(1-\eps/16)a$.  The lemma then follows from
  \[
  4(1-\eps/16)a \geq (1-\eps) \cdot 4(1 + \eps/2)a.
  \]
  It remains to consider the case $b > (\eps/2) a$.  Let $D$ be the
  disk of radius~$2a$ centered at the origin.  We have $C \subset D$,
  as otherwise $C$ contains a segment of perimeter strictly larger
  than $4a \geq 2a+2b = \peri(\Ropt)$, contradicting the optimality
  of~$\Ropt$.  It follows that the diameter of $C$ is at most~$4a$.

  By Lemma~\ref{le:square},
  $C_{\eps/16}$ contains the axis-parallel rectangle $S = [-a + t, a-t]
  \times [-b + t, b - t]$, where $t = a\eps/4$.  We have
  \[
  \peri(S) = \peri(\Ropt) - 8t
  = \peri(\Ropt) - 2a\eps 
  > (1-\eps)\peri(\Ropt). \qedhere
  \]
\end{proof}

\begin{theorem}
  Let $C$ be a convex body in the plane.  For any given $\eps\in
  (0,1)$, we can find a $(1-\eps)$-approximation to the
  maximum-perimeter rectangle contained in $C$ in time $O(\eps^{-1/2}
  T_C +\eps^{-3/2})$.
\end{theorem}
\begin{proof}
  First, we compute an $(\eps/16)$-kernel $C_{\eps/16}$ for $C$. 
  We then compute a maximum-perimeter rectangle
  contained in $C_{\eps/16}$ and return it. This finishes the
  description of the algorithm.
  
  By Lemma~\ref{le:bound3} $C_{\eps/16}$ contains a rectangle of
  perimeter at least $(1-\eps)\cdot \peri(\Ropt)$, where $\Ropt$ is a
  maximum-perimeter rectangle contained in $C$.  Therefore, the
  algorithm returns a $(1-\eps)$-approximation to the
  maximum-perimeter rectangle.

  Computing $C_{\eps/16}$ takes time $O(\eps^{-1/2} T_C)$ because of
  Lemma~\ref{le:kernel}.  Finding the maximum-perimeter rectangle
  contained in $C_{\eps/16}$ takes time $O(\eps^{-3/2})$ because of
  Theorem~\ref{th:exact2}.
\end{proof}

\begin{corollary}
  Let $C$ be a convex polygon with $n$ vertices given as a sorted
  array or a balanced binary search tree.  For any given $\eps\in
  (0,1)$, we can find a $(1-\eps)$-approximation to the
  maximum-perimeter rectangle contained in $C$ in time $O(\eps^{-1/2}
  \log n +\eps^{-3/2})$.
\end{corollary}
\begin{proof}
  In this case $T_C=O(\log n)$.
\end{proof}

%\section*{Acknowledgments}
\small 
\bibliographystyle{abbrv}
\bibliography{bibliography-rectangle}
\end{document}